\documentclass[10pt,journal]{IEEEtran}  

\IEEEoverridecommandlockouts                              

\usepackage{amsmath,graphicx,multicol}
\usepackage{amsfonts}
\usepackage{color}
\usepackage{bbm}
\usepackage{cite}
\usepackage{amssymb}
\usepackage{algorithm}
\usepackage{algorithmic}
\usepackage{tabularx}
\usepackage[utf8]{inputenc}
\usepackage[english]{babel}

\usepackage{amsthm}
\usepackage{subcaption}
\usepackage[font={footnotesize}]{caption}
\usepackage{url}
\usepackage{mathtools}
\usepackage{comment}
\usepackage[keeplastbox]{flushend}
\newtheorem{definition}{Definition}[]
\newtheorem{proposition}{Proposition}[]
\newtheorem{theorem}{Theorem}[]
\newtheorem{corollary}{Corollary}[theorem]
\newtheorem{lemma}[]{Lemma}
\newcommand\norm[1]{\left\lVert#1\right\rVert}
\newcommand{\argmax}{\arg\!\max}

\DeclareMathOperator{\E}{\mathbb{E}}

\def\x{{\mathbf x}}
\def\z{{\mathbf z}}
\def\y{{\mathbf y}}

\def\B{{\mathbf B}}
\def\H{{\mathbf A}}

\def\P{{\mathbf P}}

\def\R{{\mathbb{R}}}

\def\I{{\mathbf I}}

\def\F{{\mathbf F}}

\def\x{{\mathbf x}}
\def\w{{\mathbf w}}

\def\X{{\mathbf X}}

\def\Y{{\mathbf Y}}

\def\A{{\mathbf H}}
\def\B{{\mathbf B}}
\def\P{{\mathbf P}}

\def\R{{\mathbb{R}}}

\def\I{{\mathbf I}}
\def\a{{\mathbf h}}

\def\v{{\mathbf v}}

\newcommand{\ts}{\textsuperscript}
\newcommand{\abs}[1]{\lvert #1 \rvert}

\newcommand{\etal}{{\em et~al.~}}
\newcommand{\Var}{\mathrm{Var}}

\let\oldref\ref
\renewcommand{\ref}[1]{(\oldref{#1})}
\newcommand{\RNum}[1]{\uppercase\expandafter{\romannumeral #1\relax}}
\makeatletter
\renewcommand{\fnum@figure}{Fig.~\thefigure}
\makeatother

\title{
A Randomized Greedy Algorithm for Near-Optimal Sensor Scheduling in Large-Scale Sensor Networks
}

\author{Abolfazl~Hashemi,~\IEEEmembership{Student Member,~IEEE,} Mahsa~Ghasemi,~\IEEEmembership{Student Member,~IEEE,} Haris~Vikalo,~\IEEEmembership{Senior Member,~IEEE,} and Ufuk~Topcu,~\IEEEmembership{Member,~IEEE}
}
\begin{document}
%
\maketitle
\begin{abstract}
In a sensor network governed by a linear dynamical system, often due to practical constraints such as computational and power limitations, it is desired to select a small subset to perform the state estimation task. In this paper,  we formulate this task as the combinatorial problem of maximizing a monotone set function under a uniform matroid constraint. By introducing the notion of curvature we show that the proposed objective function is weak submodular under certain conditions by establishing an upperbound on its maximum element-wise curvature.
To efficiently solve the proposed combinatorial problem, we develop a randomized greedy algorithm that is significantly faster than state-of-the-art methods.  we analyze the performance of the proposed algorithm and establish performance guarantees on the mean square error (MSE) of the linear estimator that uses the selected sensors in terms of the optimal MSE. Extensive simulation results demonstrate efficacy of the randomized greedy algorithm in a comparison with greedy and semidefinite programming relaxation methods.
\end{abstract}
\section{Introduction}\label{sec:intro}
\IEEEPARstart{M}{odern} sensor networks, acquire myriads of measurements from a dynamical system through communication of a large number of sensors and sensor fusion centers. In these networks, due to various practical considerations and limitations on resources including computational and communication constraints, the fusion center which aggregates information typically queries only a small subset of the available sensors. This scenario, also known as the {\it sensor selection} problem, arises in various applications in control systems and signal processing including sensor selection for Kalman filtering \cite{nordio2015sensor,shamaiah2010greedy,tzoumas2016sensor}, batch state estimation and stochastic process estimation \cite{tzoumas2016scheduling,tzoumas2016near}, minimal actuator placement \cite{summers2016submodularity,tzoumas2015minimal}, voltage control and meter placement in power networks \cite{damavandi2015robust,gensollen2016submodular,liu2016towards}, sensor scheduling in wireless sensor networks \cite{shamaiah2012greedy,nordio2015sensor}, and subset selection in machine learning \cite{mirzasoleiman2014lazier}.

Although an optimal solution to the sensor selection problem can be achieved by means of branch-and-bound algorithms \cite{welch1982branch}, this requires finding solution to a computationally challenging combinatorial optimization problem which by a reduction to the set cover problem is shown to be NP-hard \cite{williamson2011design}. This in turn has motivated development of heuristics and approximate algorithms. For instance, in \cite{joshi2009sensor}, sensor selection problem is formulated as the maximization (minimization) of the $\log \det$ of the Fisher information matrix (error covariance matrix) and a semidefinite programming relaxation is proposed. The computational complexity of the SDP relaxation is cubic in the number of sensors in the network which limits practical feasibility of this scheme, especially for the modern sensor networks characterized by a growing number of sensors in the network. Additionally, the SDP relaxation does not come with any performance guarantees. To overcome these drawbacks, Shamaiah et al. \cite{shamaiah2010greedy} proposed a greedy algorithm for the $\log \det$ maximization formulation of the sensor selection problem whose complexity is  lower than that of the SDP relaxation. Since the $\log \det$ of the Fisher information matrix is a monotone submodular function, the greedy scheme in \cite{shamaiah2010greedy} is a $(1-1\slash e)$-approximation algorithm. More recently, the greedy algorithm for $\log \det$ maximization was employed and analyzed in a number of other practical settings \cite{tzoumas2016scheduling,tzoumas2015minimal,tzoumas2016near,tzoumas2016sensor}. 
All the prior work consider $\log \det$ of the Fisher information matrix which is related to the volume of the $\eta$-confidence ellipsoid. However, this criterion is not explicitly related to the mean-square error (MSE) which is often the natural performance measure of interest in sensor selection and state estimation problems. The MSE, i.e., the trace of the error covariance matrix, is not supermodular \cite{krause2008near,olshevsky2017non}. Therefore, the search for an approximation algorithm with performance guarantees on the estimator's achievable MSE remains an open research problem.

Sensor selection is related to the problem of maximizing a monotone submodular function subject to a uniform matroid constraint.
Nemhauser et al. \cite{nemhauser1978analysis}  considered this problem and showed that the greedy algorithm that iteratively selects items with maximum marginal gain provides a $(1-1/e)$-approximation factor.
In \cite{mirzasoleiman2014lazier}, a ($1-1/e-\epsilon$)-approximation stochastic-greedy algorithm is developed for the maximization of monotone increasing submodular functions under cardinality constraint that reduces the complexity of the greedy algorithm proposed in \cite{nemhauser1978analysis}. However, the assumption of submodularity in \cite{nemhauser1978analysis,mirzasoleiman2014lazier} does not hold in the sensor selection problem with MSE objective.
Recently, Wang \etal \cite{wang2016approximation} analyzed the performance of the greedy algorithm in the general setting where the function is monotone non-decreasing, but not necessarily submodular. They defined a total curvature $\mu$ and showed that the greedy algorithm provides a $(\frac{1}{1+\mu})$-approximation under matroid constraint. However, determining the elemental curvature defined in \cite{wang2016approximation} is itself an NP-hard task. Therefore, finding an explicit approximation factor for the settings where the objective function is not supermodular, e.g., trace of the error covariance matrix in sensor scheduling for state estimation via Kalman filtering, remains a challenge.

As we stated before, the natural objective function that is typically of interest in sensor selection applications, the MSE, is not submodular (or supermodular, in case one considers the minimization formulation of the problem). Hence, the performance guarantees for the greedy scheme derived in \cite{nemhauser1978analysis,fisher1978analysis} no longer hold. Moreover, processing massive amounts of data collected by modern large-scale networks may be challenging even when relying on greedy algorithms. To  address these challenges, in this paper we formulate the task of sensor selection in a large-scale sensor network as the problem of maximizing a monotone non-submodular objective function directly related to the MSE of the linear estimator of the states in a linear dynamical system. By introducing the notion of curvature $c$, we derive sufficient conditions under which the objective function of the proposed framework is weak submodular. An implication of these results is that in the important scenarios of Gaussian and Bernoulli measurement vectors that frequently come up in dimensionality-reduced Kalman filtering using random projections \cite{berberidis2016data}, the MSE objective is with high probability weak submodular. Since state-of-the-art sensor selection schemes based on greedy optimization and SDP relaxation face computational burden in modern sensor networks, we further propose a randomized greedy algorithm and find a bound on the MSE of the state estimate formed by the Kalman filter that uses the measurements of the sensors selected by the randomized greedy algorithm. Using extensive simulations on real and synthetic data, we illustrate that the proposed randomized greedy sensor selection scheme significantly outperforms both greedy and SDP relaxation methods in terms of runtime and computational complexity while providing nearly equivalent or improved performance. 
 
The rest of the paper is organized as follows. Section \oldref{sec:sys} explains the system model. In Section \oldref{sec:alg} we present the novel formulation of sensor selection problem  and establish bound on curvature of its MSE-related objective. In Section \oldref{sec:anl}, we introduce the randomized greedy algorithm and analyze its performance. Section \oldref{sec:sim} presents the simulation results while the concluding remarks are stated in Section \oldref{sec:concl}. MATLAB implementation of the proposed algorithm in this paper is freely available at \url{https://github.com/realabolfazl/RGSS}.

Before proceeding to subsequent section, we first briefly summarize the notation used in the paper. Bold capital letters refer to matrices and bold lowercase letters represent vectors. $\A_{ij}$ denotes the $(i,j)$ entry of $\A$, $\a_j$ is the $j\ts{th}$ row of $\A$, $\A_{S}$ is a submatrix of $\A$ that contains rows indexed by set $S$, and $\lambda_{max}(\A)$
and $\lambda_{min}(\A)$ are maximum and minimum eigenvalues of $\A$, respectively. Spectral ($\ell_2$) norm of a matrix is denoted by $\|.\|$. $\I_n \in \R^{n\times n}$ is the identity matrix. Moreover, let $[n] := \{1,2,\dots,n\}$.
\section{System Model and Problem Formulation}\label{sec:sys}
Consider a linear time-varying dynamical system and its measurement model,
\begin{equation}\label{eq:sys}
\begin{aligned}
\x(t+1) &= \H(t)\x(t)+\w(t) \\
\y(t) &= \A(t)\x(t)+\v(t),
\end{aligned}
\end{equation}
where $\x(t) \in \R^m$ is the state vector, $\y(t) \in 
\R^n$ is the measurement vector, $\w(t)$ and $\v(t)$ 
are zero-mean Gaussian noises with covariances
$\mathbf{Q}(t)$ and $\mathbf{R}(t)$, respectively, $\H(t) 
\in  \R^{m \times m}$ is the state transition
matrix and $\A(t) \in  \R^{n \times m}$ is the matrix 
whose rows at time $t$ are the measurement vectors 
$\a_i(t) \in \R^m$. We assume that the states $\x(t)$ are uncorrelated with $\w(t)$ and $\v(t)$. In addition, for simplicity of exposition we assume that $\x(0) \sim {\cal N}(0, \mathbf{\Sigma}_x)$, $\mathbf{Q}(t)=\sigma^2\I_m$, and $\mathbf{R}(t)=\sigma^2\I_n$.

Due to limited resources, fusion center aims to select $k$ out of $n$ sensors and use their measurements to estimate the state vector $\x(t)$ by minimizing the mean squared error (MSE) in the Kalman filtering setting. Note that we assume that the measurement vectors $\a_i(t)$ are available at the fusion center.

Let $\P_{t|t-1}$ and $\P_{t|t}$ be the prediction and filtered error covariance at
time instant $t$, respectively. Then
\begin{equation*}
	\begin{aligned}
		\P_{t|t-1} &= \H(t)\P_{t-1|t-1}\H(t)^\top+\mathbf{Q}(t) \\
		\P_{t|t} &= \left(\P_{t|t-1}^{-1}+\A_{S_{t}}(t)^\top\mathbf{R}_{S_{t}}(t)^{-1}\A_{S_{t}}(t)\right)^{-1},
	\end{aligned}
\end{equation*}
where $S_{t}$ is the set of selected sensors at time $t$ and $P_{0|0}=\mathbf{\Sigma}_x$. Since $\mathbf{R}(t)=\sigma^2\I_n$ and the measurements are uncorrelated across sensors, it holds that
\begin{equation*}
	\begin{aligned}
		\P_{t|t} = \left(\P_{t|t-1}^{-1}+\sigma^{-2}\A_{S_{t}}(t)^\top\A_{S_{t}}(t)\right)^{-1}= \F_{S_t}^{-1}
	\end{aligned}
\end{equation*}
where $\F_{S_t} = \P_{t|t-1}^{-1}+\sigma^{-2}\sum_{i\in S_{t}}\a_i(t)\a_i(t)^\top$ is the corresponding Fisher information matrix. Let
\begin{equation}\label{eq:estimator}
\begin{aligned}
\hat{\x}(t) = \frac{1}{\sigma^2}\F_{S_t}^{-1} \A_{S_{t}}(t)^\top \y(t)
\end{aligned}
\end{equation}
be the linear minimum mean-square estimator (LMMSE) of $\x(t)$. Then its MSE at time $t$ is expressed by the trace of the filtered error covariance matrix $\P_{t|t}$. That is,
\begin{equation}\label{eq:mse}
\begin{aligned}
\mathrm{MSE} = \E\left[{\norm{\x(t)-\hat{\x}_{t|t}}}_2^2\right]= \mathrm{Tr}\left(\F_{S_t}^{-1}\right)
\end{aligned}
\end{equation}
where $\hat{\x}_{t|t}$ denotes the filtered estimate of the state vector at time $t$.
To minimize MSE \ref{eq:mse} at each time step $t$ the fusion center seeks a solution to the following optimization problem:
\begin{equation}\label{eq:sensor1}
\begin{aligned}
& \underset{S}{\text{min}}
\quad \mathrm{Tr}\left(\F_S^{-1}\right)
& \text{s.t.}\hspace{0.5cm}  S \subset [n], \phantom{k}|S|=k.
\end{aligned}
\end{equation}
The combinatorial optimization problem \ref{eq:sensor1} is NP-hard by a reduction to the well-known set cover problem \cite{williamson2011design}. Intuitively, the reason is that one needs to exhaustively search over all schedules of $k$ sensors to find the optimal solution. Using the techniques established in \cite{joshi2009sensor} (although for a different optimality criterion from MSE) an approximate solution, i.e., a schedule of sensors that results in a sub-optimal MSE,  can be found by the following SDP relaxation (see Appendix 
I for the details of the derivation),
\begin{equation}\label{sdp}
\begin{aligned}
& \underset{\z,\Y}{\text{min}}
\quad \mathrm{Tr}(\Y)\\
& \text{s.t.}\hspace{0.5cm}  0\leq z_i \leq 1, \phantom{k} \forall i\in [n]\\
& \hspace{0.9cm} \sum_{i=1}^n z_i =k\\
& \hspace{0.9cm} 
\begin{bmatrix}
\Y& \I\\
\I & \P_{t|t-1}^{-1}+\sigma^{-2}\sum_{i=1}^n z_i\a_i(t)\a_i(t)^\top
\end{bmatrix} \succeq \mathbf{0}.
\end{aligned}
\end{equation}
The complexity of the SDP algorithm scales as ${\cal O}(n^3)$ which is infeasible in practice. Furthermore, there is no guarantee on the achievable MSE performance of the SDP relaxation. When the number of sensors in a network and the size of the state vector $\x(t)$ are relatively large, even the greedy algorithm proposed in \cite{shamaiah2010greedy} may be computationally prohibitive. 
\section{Proposed Formulation Based on \\Weak Submodularity}\label{sec:alg}
 In this section, we propose a new formulation for optimizing MSE in a sensor selection task for state estimation via Kalman filtering in a sensor network by leveraging the idea of {\it weak submodularity}. First, we overview some definitions that are essential in the development of the proposed framework.
\begin{definition}
	\label{def:unif}
	A set function $f:2^X\rightarrow \mathbb{R}$ is monotone non-decreasing if $f(S)\leq f(T)$ for all $S\subseteq T\subseteq X$.
\end{definition}
\begin{definition}
	\label{def:submod}
	A set function $f:2^X\rightarrow \mathbb{R}$ is submodular if 
	\begin{equation}
	f(S\cup \{j\})-f(S) \geq f(T\cup \{j\})-f(T)
	\end{equation}
	for all subsets $S\subseteq T\subset X$ and $j\in X\backslash T$. The term $f_j(S)=f(S\cup \{j\})-f(S)$ is the marginal value of adding element $j$ to set $S$.
\end{definition}

A closely related concept to submodularity is the notion of curvature of a set function that quantifies how close the function is to being submodular. Here, we define the element-wise curvature.
\begin{definition}
	The element-wise curvature of a monotone non-decreasing function $f$ is defined as
	\begin{equation}
	{\cal C}_l=\max_{(S,T,i)\in \mathcal{X}_l}{f_i(T)\slash f_i(S)},
	\end{equation}
	where $\mathcal{X}_l = \{(S,T,i)|S \subset T \subset X, i\in X \backslash T, |T\backslash S|=l,|X|=n\}$. Furthermore, the maximum element-wise curvature is given by ${\cal C}_{\max}=\max_{l=1}^{n-1}{{\cal C}_l}$.
\end{definition}
When  ${\cal C}_{\max} > 1$, $f(S)$ is called a weak submodular set function. Note that a set function is submodular if and only if ${\cal C}_{\max} \le 1$. Further, we say $f(S)$ is weak submodular iff it has a bounded ${\cal C}_{\max}$.

\begin{definition}
	Let $X$ be a finite set and let $\mathcal{I}$ be a collection of subsets of $X$. The pair $\mathcal{M}=(X,\mathcal{I})$ is a matroid if the following properties hold: 
	\begin{itemize}
		\item Hereditary property. If $T\in \mathcal{I}$, then $S\in \mathcal{I}$ for all $S\subseteq T$.
		\item Augmentation property. If $S,T\in \mathcal{I}$ and $\abs{S}<\abs{T}$, then there exists $e\in T\backslash S$ such that $S\cup \{e\}\in \mathcal{I}$.
	\end{itemize}
	The collection $\mathcal{I}$ is called the set of independent sets of the matroid $\mathcal{M}$. A maximal independent set is a basis. It is easy to show that all the bases of a matroid have the same cardinality.
\end{definition}

Given a monotone non-decreasing set function $f:2^X\rightarrow \mathbb{R}$ with $f(\emptyset)=0$, and a uniform matroid $\mathcal{M}=(X,\mathcal{I})$, we
are interested in the combinatorial problem
\begin{equation}\label{eq:fmax}
\max_{S \in \mathcal{I}} f(S).
\end{equation}

Next, we establish the proposed framework. Let 
\[f(S) = \mathrm{Tr}\left(\P_{t|t-1}-\F_S^{-1}\right). \]
Evidently, since $\P_{t|t-1}$ is known, given the value of $f(S)$ one can easily infer the corresponding MSE of linear estimator using subset $S$ of sensors selected at time $t$. Then, we can express the optimization problem in \ref{eq:sensor1} as
\begin{equation}\label{eq:sensor}
\begin{aligned}
& \underset{S}{\text{max}}
\quad f(S)
& \text{s.t.}\hspace{0.5cm}  S \subset [n], \phantom{k}|S|=k.
\end{aligned}
\end{equation}
We now argue that \ref{eq:sensor} is indeed an instance of the general combinatorial problem \ref{eq:fmax}.
By defining $X = [n]$ and $\mathcal{I} = \{S \subset X| |S|=k\}$, it is easy to see that $\mathcal{M}=(X,\mathcal{I})$ is a matroid.  In Proposition \oldref{p:mono} below we characterize important properties of $f(S)$ and develop a recursive scheme to efficiently compute the marginal gain of querying a sensor. The formula for the marginal gain of $f(S)$ is also of interest in our subsequent analysis of its weak submodularity properties.
\begin{proposition}\label{p:mono}
\textit{Let $f(S) = \mathrm{Tr}\left(\P_{t|t-1}-\F_S^{-1}\right)$. Then, $f(S)$ is a monotonically increasing set function, $f(\emptyset)=0$, and 
\begin{equation}\label{eq:mg}
f_j(S) = \frac{\a_j(t)^\top\F_S^{-2}\a_j(t)}{\sigma^{2}+\a_j(t)^\top\F_S^{-1}\a_j(t)},
\end{equation}
where,
\begin{equation}\label{eq:upf}
\F_{S \cup\{j\}}^{-1} = \F_{S}^{-1}-\frac{\F_{S}^{-1}\a_{j}(t)\a_{j}(t)^\top\F_{S}^{-1}}{\sigma^2+\a_{j}(t)^\top\F_{S}^{-1}\a_{j}(t)}.
\end{equation}}
\end{proposition}
\begin{proof}
See Appendix 
II.
\end{proof}

Recall, as we stated it is shown that MSE is not supermodular \cite{olshevsky2017non,krause2008near}. This immediately implies that the proposed objective $f(S) = \mathrm{Tr}\left(\P_{t|t-1}-\F_S^{-1}\right)$ is not submodular as it is the additive inverse of MSE. However as we show in Theorem \oldref{thm:curv}, under certain conditions, $f(S)$ is characterized with a bounded maximum element-wise curvature $\mathcal{C}_{\max}$. Theorem \oldref{thm:curv} also states a probabilistic theoretical upper bound on $\mathcal{C}_{\max}$ in scenarios where at each time step the measurement vectors $\a_j(t)$'s are i.i.d. random vectors.

Before proceeding to Theorem \oldref{thm:curv} and its proof, we first state the  matrix Bernstein inequality \cite{tropp2015introduction} that will be used in the proof of Theorem \oldref{thm:curv}.
\begin{lemma}\label{thm:ber}
	(Theorem 6.6.1 in \cite{tropp2015introduction}.) Let $\{\X_\ell\}_{\ell=1}^n$ be a finite collection of independent, random, Hermitian matrices in $\R^{m\times m}$. Assume that for all $\ell \in [n]$,
	\begin{equation}\label{eq:ber}
	\E\left[\X_\ell\right] = \mathbf{0}, \quad \lambda_{max}(\X_\ell)\leq L.
	\end{equation}
	Let $\Y=\sum_{\ell=1}^n \X_\ell$. Then, for all $q>0$, it holds that
	\begin{equation}
	\Pr\{\lambda_{max}(\Y)\geq q \} \leq m\exp\left(\frac{-q^2\slash 2}{\|\E\left[\Y^2\right]\|+Lq\slash 3}\right).
	\end{equation} 
\end{lemma}
We now proceed to the statement and proof of Theorem \oldref{thm:curv}.
\begin{theorem}\label{thm:curv}
	Let $\mathcal{C}_{max}$ be the maximum element-wise curvature of $f(S)$, i.e., the objective function of sensor scheduling problem and assume $\|\a_j(t)\|_2^2\leq C$ for all $j$ and $t$. Then, if 
	\begin{equation}\label{eq:fstcond}
		\frac{1}{\sigma^2}\lambda_{max}(\A(t)^\top\A(t)) \leq \frac{1}{\phi} -\frac{1}{\lambda_{min}(\P_{t|t-1})}
	\end{equation}
	for some $0<\phi< \lambda_{min}(\P_{t|t-1})$, it holds that 
	\begin{equation}\label{eq:phic}
	\mathcal{C}_{max} \leq \frac{\lambda_{max}(\P_{t|t-1})^{2}(\sigma^{2}+\lambda_{max}(\P_{t|t-1})C)}{\phi^{2}(\sigma^{2}+\phi C)},
	\end{equation}
	Furthermore, if $\a_j(t)$'s are independent zero-mean random vectors with covariance matrix $\sigma_h^2\I_m$ such that for all $j$, $\sigma_h^2<C$, for all $q>0$ with probability 
	\begin{equation} \label{eq:psuc}
	p\geq 1-m\exp\left(\frac{-q^2\slash 2}{(C-\sigma_h^2)(n\sigma_h^2+q\slash 3)}\right)
	\end{equation}
    it holds that
	\begin{equation}\label{eq:phi}
	\phi \geq \left(\frac{1}{\lambda_{min}(\P_{t|t-1})}+\frac{n\sigma_h^2+q}{\sigma^2}\right)^{-1}.
	\end{equation}
\end{theorem}
\begin{proof}
	First, from the definition of element-wise curvature and \ref{eq:mg} we obtain that
	\begin{equation}
	\begin{aligned}
	{\cal C}_l&=\max_{(S,T,j)\in \mathcal{X}_l}{\frac{(\a_j(t)^\top\F_T^{-2}\a_j(t))(\sigma^{2}+\a_j(t)^\top\F_S^{-1}\a_j(t))}{(\a_j(t)^\top\F_S^{-2}\a_j(t))(\sigma^{2}+\a_j(t)^\top\F_T^{-1}\a_j(t))}}\\\vspace{0.2cm}
	&\leq \max_{(S,T,j)\in \mathcal{X}_l}{\frac{\lambda_{max}(\F_T^{-2})(\sigma^{2}+\lambda_{max}(\F_S^{-1})\|\a_j(t)\|_2^2)}{\lambda_{min}(\F_S^{-2})(\sigma^{2}+\lambda_{min}(\F_T^{-1})\|\a_j(t)\|_2^2)}},
	\end{aligned}
	\end{equation}
	where the last inequality follows from the Courant–Fischer min-max theorem \cite{bellman1997introduction}. Notice that $\lambda_{max}(\F_S^{-1})=\lambda_{min}(\F_S)^{-1}$ and $\lambda_{min}(\F_T)\geq\lambda_{min}(\F_S)\geq\lambda_{min}(\F_\emptyset)=\lambda_{min}(\P_{t|t-1}^{-1})$. This fact, along with the definition of $\mathcal{C}_{max}$ implies
	\begin{equation}\label{eq:ray}
	\begin{aligned}
	\mathcal{C}_{max} &\leq \frac{\lambda_{max}(\P_{t|t-1})^{2}(\sigma^{2}+\lambda_{max}(\P_{t|t-1})\|\a_j(t)\|_2^2)}{\lambda_{max}(\F_S)^{-2}(\sigma^{2}+\lambda_{max}(\F_T)^{-1}\|\a_j(t)\|_2^2)}\\\vspace{0.2cm}
	&\stackrel{(a)}{\leq} \frac{\lambda_{max}(\P_{t|t-1})^{2}(\sigma^{2}+\lambda_{max}(\P_{t|t-1})\|\a_j(t)\|_2^2)}{\lambda_{max}(\F_{[n]})^{-2}(\sigma^{2}+\lambda_{max}(\F_{[n]})^{-1}\|\a_j(t)\|_2^2)}\\\vspace{0.2cm}
	&\stackrel{(b)}{\leq} \frac{\lambda_{max}(\P_{t|t-1})^{2}(\sigma^{2}+\lambda_{max}(\P_{t|t-1})C)}{\lambda_{max}(\F_{[n]})^{-2}(\sigma^{2}+\lambda_{max}(\F_{[n]})^{-1}C)},
	\end{aligned}
	\end{equation}
	where (a) follows from the fact that $\lambda_{max}(\F_S)\leq\lambda_{max}(\F_T)\leq\lambda_{max}(\F_{[n]})$ and (b) holds since
	\begin{equation}
	g(x) = \frac{\sigma^{2}+\lambda_{max}(\P_{t|t-1})x}{\sigma^{2}+\lambda_{max}(\F_{[n]})^{-1}x}
	\end{equation} 
	is a monotonically increasing function for $x>0$. Now, since the maximum eigenvalue of a positive definite matrix satisfies the triangle inequality, we have
	\begin{equation}
	\begin{aligned}
	\lambda_{max}(\F_{[n]})&\leq \frac{1}{\lambda_{min}(\P_{t|t-1})}+\frac{1}{\sigma^2}\lambda_{max}(\sum_{j=1}^n \a_j(t)\a_j(t)^\top)\\
	&\leq\frac{1}{\lambda_{min}(\P_{t|t-1})}+\frac{1}{\sigma^2}\lambda_{max}(\A(t)^\top\A(t))
	\end{aligned}
	\end{equation}
	Hence, by combining the condition \ref{eq:fstcond} and \ref{eq:ray}, we obtain the results stated in \ref{eq:phic}.
	Next, to obtain the expression for the scenario of i.i.d random measurement vectors we now bound $\lambda_{max}(\F_{[n]})$ using the matrix Bernstein inequality \cite{tropp2015introduction}.
	Let $\X_j=\a_j(t)\a_j(t)^\top-\sigma_h^2\I_m$ and $\Y=\sum_{j=1}^n \X_j$. To use the result of Lemma \oldref{thm:ber}, one should find the quantities in \ref{eq:ber}. Note that,
	\begin{equation}
	\begin{aligned}
	\E[\X_j]&= \E[\a_j(t)\a_j(t)^\top-\sigma_h^2\I_m] \\
	&=\E[\a_j(t)\a_j(t)^\top] -\sigma_h^2\I_m =\mathbf{0}. 
	\end{aligned}
	\end{equation}
	This  in turn  implies that $\E[\Y]=\mathbf{0}$.
	Since $\X_j$'s are independent,
	\begin{equation}
	\|\E[\Y^2]\|= \|\E[\sum_{j=1}^n \X_j^2]\|\leq \sum_{j=1}^n \|\E[\X_j^2]\|
	\end{equation}
	by the linearity of expectation and triangle inequality. It just remains to determine
	$\lambda_{max}(\X_j)$ and $\E[\X_j^2]$.
	
	First, we verify $\a_j$ is an eigenvector of $\X_j$:
	\begin{equation}
	\begin{aligned}
	\X_j\a_j &= \left(\a_j(t)\a_j(t)^\top-\sigma_h^2\I_m\right)\a_j\\
	& = \left(\|\a_j(t)\|_2^2-\sigma_h^2\right)\a_j.
	\end{aligned}
	\end{equation} 
	where $\a_j(t)\a_j(t)^\top-\sigma_h^2\I_m$ is the corresponding eigenvalue. Since $\a_j(t)\a_j(t)^\top$ is a rank-1 matrix, other eigenvalues of $\X_j$ are all equal to $-\sigma_h^2$. Hence, 
	\begin{equation}
	\lambda_{max}(\X_j) \leq C-\sigma_h^2 >0. 
	\end{equation}
	
	We now establish an upper-bound for $\E[\X_j^2]$ as follows:
	\begin{equation}
	\begin{aligned}
	\E[\X_j^2] &= \E[\left(\a_j(t)\a_j(t)^\top-\sigma_h^2\I_m\right)\left(\a_j(t)\a_j(t)^\top-\sigma_h^2\I_m\right)]\\
	& = \left(\|\a_j(t)\|_2^2-\sigma_h^2\right)\E[ \a_j(t)\a_j(t)^\top] \\
	&\qquad \qquad \qquad - \sigma_h^2\E[ \left(\a_j(t)\a_j(t)^\top-\sigma_h^2\I_m\right)]\\
	& = \left(\|\a_j(t)\|_2^2-\sigma_h^2\right)\sigma_h^2\I_m\preceq (C-\sigma_h^2)\sigma_h^2\I_m
	\end{aligned}
	\end{equation}
	where we have used the fact that $\E[\X_j]=\mathbf{0}$. Thus, $L =C-\sigma_h^2$ and $\|\E[\Y^2]\| \leq n(C-\sigma_h^2)\sigma_h^2$. Now, according to Lemma \oldref{thm:ber}, for all $q>0$ it holds that $\Pr\{\lambda_{max}(\Y)\leq q\}\geq p$ where 
	\begin{equation}\label{eq:prob}
	p= 1-m\exp\left(\frac{-q^2\slash 2}{(C-\sigma_h^2)(n\sigma_h^2+q\slash 3)}\right).
	\end{equation}
	Therefore, 
	\begin{equation}
	\lambda_{max}(\F_{[n]})\leq \frac{1}{\lambda_{min}(\P_{t|t-1})}+\frac{n\sigma_h^2+q}{\sigma^2}= \phi^{-1}
	\end{equation}
	with probability $p$. This completes the proof.
\end{proof}
\textit{Remark 1:} The setting of i.i.d. random vectors described in Theorem \oldref{thm:curv} arises in scenarios where sketching techniques, such as random projections are used to reduce dimensionality of the measurement equation (see \cite{berberidis2016data} for more details). Such sketching schemes give rise to the following important and widely-used examples:
\begin{enumerate}
	\item {\it Multivariate Gaussian measurement vectors:} Let $\a_j(t) \sim \mathcal{N}(0,\frac{1}{m}\I_m)$ for all $j$. It is easy to show that $\E[\|\a_j(t)\|_2^2] = 1$ for all $j$. Furthermore, it can be shown that $\|\a_j(t)\|_2^2$ is with high probability distributed around its expected value. Therefore, for this case,  $\sigma_h^2 = \frac{1}{m}$ and $C = 1$.
	\item {\it Centered Bernoulli measurement vectors:} Let each entry of $\a_j(t)$ be set to $\pm \frac{1}{\sqrt{m}}$ with equal probability. Therefore, $\|\a_j(t)\|_2^2=1 = C $. Additionally, $\sigma_h^2 = \frac{1}{m}$ since the entries of $\a_j(t)$ are i.i.d. zero-mean random variables with variance $\frac{1}{m}$.
\end{enumerate}

The conditions stated in Theorem \oldref{thm:curv} can be interpreted as conditions on the condition number of  $\P_{t|t-1}$ as explained next.  For sufficiently large $m$, and when $\sigma_h^2 = \frac{1}{m}$, we can approximate $C \approx 1$. Assume $\phi \geq \lambda_{max}(\P_{t|t-1})\slash \Delta$ for some $\Delta>1$. Define
\begin{equation}
\mathrm{SNR} = \frac{\lambda_{max}(\P_{t|t-1})}{\sigma^2},  
\end{equation}
and let
\begin{equation}
\kappa = \frac{\lambda_{max}(\P_{t|t-1})}{\lambda_{min}(\P_{t|t-1})} \geq 1 
\end{equation}
be the condition number of $\P_{t|t-1}$. Then, with some elementary numerical approximations we obtain the following corollary.
\begin{corollary}\label{col:con}
	Assume 
	\begin{equation}
	\Delta \geq \kappa +c_1\frac{n}{m}\mathrm{SNR},
	\end{equation}
	for some $c1>1$. Then, with probability 
	\begin{equation}
	p\geq 1-m\exp(-\frac{n}{m}c_2),
	\end{equation}
	it holds that $\mathcal{C}_{max} \leq \Delta^3$ for some $c_2>0$.
\end{corollary}
Hence, informally, Theorem \oldref{thm:curv} states that for a well-conditioned $\P_{t|t-1}$, curvature of $f(S)$ is small, implying weak-submodularity of $f(S)$. Furthermore, the probability of such event is exponentially increasing in the number of available measurements. 
\section{Randomized Greedy Sensor Selection}\label{sec:anl}
\renewcommand\algorithmicdo{}	
\begin{algorithm}[t]
	\caption{Randomized Greedy Sensor Scheduling}
	\label{alg:greedy}
	\begin{algorithmic}[1]
		\STATE \textbf{Input:}  $\P_{t|t-1}$, $\A_t$, $k$, $\epsilon$.
		\STATE \textbf{Output:} Subset $S_t\subseteq [n]$ with $|S_t|=k$.
		\STATE Initialize $S_t^{(0)} =  \emptyset$, $\F_{S_t^{(0)}}=\P_{t|t-1}^{-1}$.
		\FOR{$i = 0,\dots, k-1$}
		\STATE Choose $R$ by sampling $s=\frac{n}{k}\log{(1/\epsilon)}$ indices uniformly at random from $[n]\backslash S_t^{(i)}$.\vspace{0.2cm}
		\STATE $i_s = \argmax_{j\in R} \frac{\a_j(t)^\top\F_{S_t^{(i)}}^{-2}\a_j(t)}{\sigma^{2}+\a_j(t)^\top\F_{S_t^{(i)}}^{-1}\a_j(t)}$.\vspace{0.2cm}
		\STATE Set $S_t^{(i+1)}= S_t^{(i)}\cup \{i_s\}$.\vspace{0.2cm}
		\STATE $\F_{S_t^{(i+1)}}^{-1} = \F_{S_t^{(i)}}^{-1}-\frac{\F_{S_t^{(i)}}^{-1}\a_{i_s}(t)\a_{i_s}(t)^\top\F_{S_t^{(i)}}^{-1}}{\sigma^2+\a_{i_s}(t)^\top\F_{S_t^{(i)}}^{-1}\a_{i_s}(t)}$
		\ENDFOR
		\RETURN $S_t = S_t^{(k)}$.
	\end{algorithmic}
\end{algorithm}
In the this section, we present a randomized greedy algorithm to approximately solve optimization problem \ref{eq:sensor} and provide its performance guarantees. 

Given prohibitive complexity of SDP relaxation and greedy schemes for sensor scheduling in large-scale systems, to provide practical feasibility, inspired by the algorithm developed in \cite{mirzasoleiman2014lazier} that only works for submodular objectives,  we propose a computationally efficient randomized greedy algorithm (see Algorithm \oldref{alg:greedy}) that finds an approximate solution to \ref{eq:sensor} with a guarantee on its achievable MSE. Algorithm \oldref{alg:greedy} performs the task of sensor scheduling in the following way. At each iteration of the algorithm, a subset $R$ of size $s$ is sampled uniformly at random and without replacement from the set of sensors. The marginal gain provided by each of these $s$ sensors to the objective function is computed using \ref{eq:mg}, and the one yielding the highest marginal gain is added to the set of selected sensors. Then the efficient recursive formula in \ref{eq:upf} is used to update $\F_S^{-1}$ so it can be used in the next iteration. This procedure is repeated $k$ times.

\textit{Remark 2:} The parameter $\epsilon$ in Algorithm \oldref{alg:greedy}, $e^{-k}\leq\epsilon<1$, denotes a predefined constant that is chosen to strike a desired balance between performance and complexity. When $\epsilon = e^{-k}$, each iteration includes all of the non-selected sensors in $R$ and Algorithm \oldref{alg:greedy} coincides with the greedy scheme. However, as $\epsilon$ approaches $1$, $|R|$ and thus the overall computational complexity decreases. 
\subsection{Performance Analysis of the Proposed Scheme}
In this section we analyze performance and complexity of Algorithm \oldref{alg:greedy} and in Theorem \oldref{thm:card} provide a bound on the performance of the proposed randomized greedy scheme when applied to finding an approximate solution to the maximization \ref{eq:sensor}.

Before stating the main results, we first provide two lemmas. Lemma \oldref{lem:curv} upper-bounds the difference between the values of the objective corresponding to two sets having different cardinalities while Lemma \oldref{lem:rand} provides a lower bound on the expected marginal gain.

\begin{lemma}\label{lem:curv}
\textit{Let $\{{\cal C}_l\}_{l=1}^{n-1}$ be the element-wise curvatures of $f(S)$. Let $S$ and $T$ be any schedules of sensors such that $S\subset T \subseteq [n]$ with $|T\backslash S|=r$. Then, it holds that
\begin{equation}
f(T)-f(S)\leq  C(r)\sum_{j\in T\backslash S}f_j(S),
\end{equation}
where $ C(r)=\frac{1}{r}(1+\sum_{l=1}^{r-1}{\cal C}_l)$.}
\end{lemma}
\begin{proof}
See Appendix 
III.
\end{proof}
\begin{lemma}\label{lem:rand}
\textit{Let $S_t^{(i)}$ be the set of selected sensors at the end of the $i\ts{th}$ iteration of Algorithm \oldref{alg:greedy}. Then
\begin{equation}
\E\left[f_{(i+1)_s}(S_t^{(i)})|S_t^{(i)}\right]\geq \frac{1-\epsilon^{\beta}}{k}\sum_{j\in O_t\backslash S_t^{(i)}}f_j(S_t^{(i)}),
\end{equation}
where $O_t$ is the set of optimal sensors at time $t$, $i_s$ is the index of the selected sensor at the $i\ts{th}$ iteration, $\beta=1+\max\{0,\frac{s}{2n}-\frac{1}{2(n-s)}\}$, and $s=\frac{n}{k}\log{(1/\epsilon)}$.}
\end{lemma}
\begin{proof}
See Appendix 
IV.
\end{proof}
Theorem \oldref{thm:card} below states that Algorithm \oldref{alg:greedy} provides an approximate solution to the sensor scheduling problem. In particular, if $f(S)$ is characterized by a bounded maximum element-wise curvature, Algorithm \oldref{alg:greedy} returns a subset of sensors yielding an objective that is on average within a multiplicative factor of the objective achieved by the optimal schedule.

\begin{theorem}\label{thm:card}
\textit{Let $\mathcal{C}_{max}$ be the maximum element-wise curvature of $f(S)$, i.e., the objective function of sensor scheduling problem in \ref{eq:sensor}. Let $S_t$ denote the schedule of sensors selected by Algorithm \oldref{alg:greedy} at time $t$, and let $O_t$ be the optimum solution of \ref{eq:sensor} such that $|O_t|=k$. Then $f(S_t)$ is on expectation a multiplicative factor away from $f(O_t)$. That is,}
\begin{equation}\label{eq:card}
\E\left[f(S_t)\right]\geq \left(1-e^{-\frac{1}{c}}-\frac{\epsilon^\beta}{c}\right) f(O_t),
\end{equation}
\textit{where $c=\max\{{\cal C}_{\max},1\}$, $e^{-k}\leq\epsilon<1$, and $\beta=1+\max\{0,\frac{s}{2n}-\frac{1}{2(n-s)}\}$. Furthermore, the computational complexity of Algorithm \oldref{alg:greedy} is ${\cal O}(nm^2\log(\frac{1}{\epsilon}))$ where $n$ is the total number of sensors and $m$ is the dimension of $\x_t$.}
\end{theorem}
\begin{proof}
Consider $S_t^{(i)}$, the set generated at the end of the $i\ts{th}$ iteration of Algorithm \oldref{alg:greedy}. 
Employing Lemma \oldref{lem:curv} with $S=S_t^{(i)}$ and $T=O_t\cup S_t^{(i)}$, and using monotonicity of $f$ yields
\begin{equation}
\begin{aligned}
\frac{f(O_t)-f(S_t^{(i)})}{\frac{1}{r}\left(1+\sum_{l=1}^{r-1}{\cal C}_l\right)}&\leq \frac{f(O_t\cup S_t^{(i)})-f(S_t^{(i)})}{\frac{1}{r}\left(1+\sum_{l=1}^{r-1}{\cal C}_l\right)}\\
&\leq  \sum_{j\in O_t\backslash S_t^{(i)}}f_j(S_t^{(i)}),
\end{aligned}
\end{equation}
where $|O_t\backslash S_t^{(i)}|=r$. Now, using Lemma \oldref{lem:rand} we obtain
\begin{equation}
\E\left[f_{(i+1)_s}(S_t^{(i)})|S_t^{(i)}\right]\geq \left(1-\epsilon^{\beta}\right)\frac{f(O_t)-f(S_t^{(i)})}{\frac{k}{r}\left(1+\sum_{l=1}^{r-1}{\cal C}_l\right)}.
\end{equation}
Applying the law of total expectation yields
\begin{equation}
\begin{aligned}
\E\left[f_{(i+1)_s}(S_t^{(i)})\right]&=\E\left[f(S_t^{(i+1)})-f(S_t^{(i)})\right]\\
&\geq \left(1-\epsilon^{\beta}\right)\frac{f(O_t)-\E\left[f(S_t^{(i)})\right]}{\frac{k}{r}\left(1+\sum_{l=1}^{r-1}{\cal C}_l\right)}.
\end{aligned}
\end{equation}
Using the definition of the maximum element-wise curvature, we obtain 
\begin{equation}
\frac{1}{r}\left(1+\sum_{l=1}^{r-1}{\cal C}_l\right)\leq \frac{1}{r}(1+(r-1){\cal C}_{\max}) =g(r).
\end{equation}
It is easy to verify, e.g., by taking the derivative, that $g(r)$ is decreasing (increasing) with respect to $r$ if ${\cal C}_{\max}<1$ (${\cal C}_{\max}>1$). Let $c=\max\{{\cal C}_{\max},1\}$. Then 
\begin{equation}
\frac{1}{r} \left(1+\sum_{l=1}^{r-1}{\cal C}_l\right)\leq \frac{1}{r}(1+(r-1){\cal C}_{\max}) \leq c.
\end{equation}
Hence, 
\begin{equation}
\E\left[f(S_t^{(i+1)})-f(S_t^{(i)})\right]\geq \frac{1-\epsilon^\beta}{kc}\left(f(O_t)-\E\left[f(S_t^{(i)})\right]\right).
\end{equation}
By induction and due to the fact that $f(\emptyset) = 0$, 
\begin{equation}
\E[f(S_t)]\geq \left(1-\left(1-\frac{1-\epsilon^\beta}{kc}\right)^k\right)f(O_t).
\end{equation}
Finally, using the fact that $(1+x)^y\leq e^{xy}$ for $y>0$ and  the easily verifiable fact that $e^{ax}\leq 1+axe^a$ for $0<x<1$, 
\begin{equation}
\begin{aligned}
\E[f(S_t)]&\geq \left(1-e^{-\frac{1-\epsilon^\beta}{c}}\right)f(O_t)\\
&\stackrel{}{\geq} \left(1-e^{-\frac{1}{c}}-\frac{\epsilon^\beta}{c}\right)f(O_t).
\end{aligned}
\end{equation}
To take a closer look at computational complexity, note that step 6 costs $\mathcal{O}(\frac{n}{k}m^2\log(\frac{1}{\epsilon}))$ as one needs to compute $\frac{n}{k}\log(\frac{1}{\epsilon})$ marginal gains, each with complexity ${\cal O}(m^2)$. Step 8 requires ${\cal O}(m^2)$ arithmetic operations. Since there are $k$ such iterations, running time of Algorithm \oldref{alg:greedy} is ${\cal O}(nm^2\log(\frac{1}{\epsilon}))$. This completes the proof.
\end{proof}
Using the definition of $f(S)$ we obtain Corollary \oldref{col:mse} stating that, at each time step, the achievable mean-square error in \ref{eq:mse} obtained by forming an estimate using sensors selected by the randomized greedy algorithm is within a factor of the optimal mean-square error.

\begin{corollary}\label{col:mse}
	Instate the notation and hypothesis of Theorem \oldref{thm:card} and let $\alpha =1-e^{-\frac{1}{c}}-\frac{\epsilon}{c}$. Let $\mathrm{MSE}_{S_t}$ denote the mean-square estimation error obtained by forming an estimate using information provided by the sensors selected by Algorithm \oldref{alg:greedy} at time $t$, and let $\mathrm{MSE}_{o}$ be the optimal mean-square error formed using information collected by optimum solution of \ref{eq:sensor}. Then the expected $\mathrm{MSE}_{S_t}$ is bounded as
\begin{equation}\label{eq:card1}
\E\left[\mathrm{MSE}_{S_t}\right]\leq \alpha \mathrm{MSE}_{o} + (1-\alpha) \mathrm{Tr}(\P_{t|t-1}).
\end{equation}
\end{corollary}

\textit{Remark 3:} Since the proposed scheme is a randomized algorithm, Theorem \oldref{thm:card} and Corollary \oldref{col:mse} state that the expected MSE associated with the solution returned by Algorithm \oldref{alg:greedy} is a multiplicative factor $\alpha$ away from the optimal MSE. Notice that, as we expect, $\alpha$ is decreasing in both $c$ and $\epsilon$. If $f(S)$ is characterized by a small curvature, then $f(S)$ is nearly submodular and randomized greedy algorithm delivers a near-optimal scheduling. As we decrease $\epsilon$, $\alpha$ increases which in turn results in a better approximation factor. In the limit, if $\epsilon = e^{-k}$, then $\alpha = 1-e^{-\frac{1}{c}}-\frac{e^{-k}}{c}$ corresponds to the approximation factor of the greedy algorithm. Notice that the negligible term $-\frac{e^{-k}}{c}$ stems from the specific analysis that we employed to treat the randomization step of Algorithm 1. In fact, one can show $\alpha = 1-e^{-\frac{1}{c}}$ for the greedy algorithm by following a similar argument as that of the classical analysis given in \cite{nemhauser1978analysis}.

\textit{Remark 4:} The computational complexity of the greedy method for sensor selection that utilizes the efficient recursions given in Proposition \ref{p:mono} to find the marginal gains  is ${\cal O}(knm^2)$. Hence, our proposed scheme provides a reduction in complexity by $k\slash \log(\frac{1}{\epsilon})$ which may be particularly beneficial in large-scale networks, as we illustrate in our simulation results. 

Next, we study the performance of the randomized greedy algorithm using the tools of probably approximately correct (PAC) learning theory \cite{valiant1984theory}. The randomized selection step of Algorithm \oldref{alg:greedy} can be interpreted as approximating the marginal gains of the selected sensors using a  greedy scheme \cite{shamaiah2010greedy}. More specifically, for the $i^{th}$ iteration it holds that $f_{j_{rg}}(S_t^{(i)}) = \eta_t^{(i)} f_{j_{g}}(S_t^{(i)})$, where subscripts $rg$ and $g$ refer to the sensors selected by the randomized greedy (Algorithm \oldref{alg:greedy}) and the greedy algorithm, respectively, and $\ell_{i}(\epsilon)\leq\eta_t^{(i)}\leq 1$ for all $i \in [k]$ are random variables with mean $\mu_i(\epsilon)$. \footnote{Notice that $\ell_{i}(\epsilon)$ and $\mu_{i}(\epsilon)$ are time-varying quantities where the time index is omitted for simplicity of the notation.}
In view of this argument, we obtain 
Theorem \oldref{thm:pac} which states that  if $f(S)$ is characterized by a bounded maximum element-wise curvature and $\{\eta_t^{(i)}\}_{i=1}^k$ are independent random variables, Algorithm \oldref{alg:greedy} returns a subset of sensors yielding an objective that with high probability is only a multiplicative factor away from the objective achieved by the optimal schedule.
\begin{theorem}\label{thm:pac}
	Instate the notation and hypotheses of Theorem \oldref{thm:card}. Assume $\{\eta_t^{(i)}\}_{i=1}^k$ is a collection of random variables such that $\ell_{i}(\epsilon)\leq\eta_t^{(i)}\leq 1$, and $\E[\eta_t^{(i)}] = \mu_i(\epsilon)$ for all $i$ and $t$. let $\ell_{min}(\epsilon) = \min_{i,t}\{\ell_i(\epsilon)\}$ and $\mu_{min}(\epsilon) = \min_{i,t}\{\mu_i(\epsilon)\}$. Then
	\begin{equation}
	\begin{aligned}
	f(S_t) \geq \left(1- e^{-\frac{\ell_{min}(\epsilon)}{c}}\right)f(O_t),
	\end{aligned}
	\end{equation}
	Furthermore, if  $\{\eta_t^{(i)}\}_{i=1}^k$ are independent, for all $0<q<1$, with probability at least $1-e^{-Ck}$ it holds that
	\begin{equation}\label{eq:pacbound}
f(S_t)\geq \left(1-e^{-\frac{(1-q)\mu_{min}(\epsilon)}{c}}\right) f(O_t),
	\end{equation}
	for some $C>0$.
\end{theorem}
\begin{proof}
Consider $S_t^{(i)}$, the set generated at the end of the $i\ts{th}$ iteration of Algorithm \oldref{alg:greedy} and let $(i+1)_{g}$ and $(i+1)_{rg}$ denote the sensors selected by greedy and randomized greedy algorithm at $i\ts{th}$ iteration, respectively.  Let $c=\max\{{\cal C}_{\max},1\}$. Employing Lemma \oldref{lem:curv} with $S=S_t^{(i)}$ and $T=O_t\cup S_t^{(i)}$, and using monotonicity of $f$ yields
\begin{equation}
\begin{aligned}
f(O_t)-f(S_t^{(i)})&\leq f(O_t\cup S_t^{(i)})-f(S_t^{(i)})\\
&\leq  c \sum_{j\in O_t\backslash S_t^{(i)}}f_j(S_t^{(i)}).
\end{aligned}
\end{equation}
Using the fact that 
\begin{equation}
f_j(S_t^{(i)})\leq f_{(i+1)_{rg}}(S_t^{(i)}) \leq f_{(i+1)_{g}}(S_t^{(i)})
\end{equation}
for all $j$, we obtain
\begin{equation}\label{eq:rel1}
\begin{aligned}
f(O_t)-f(S_t^{(i)})\leq ckf_{(i+1)_{g}}(S_t^{(i)}).
\end{aligned}
\end{equation}
On the other hand,
\begin{equation}\label{eq:rel2}
\begin{aligned}
f(S_t^{(i+1)}) - f(S_t^{(i)})&=f_{(i+1)_{rg}}(S_t^{(i)}) \\
&=\eta_t^{(i+1)} f_{(i+1)_{g}}(S_t^{(i)}).
\end{aligned}
\end{equation} 
Combining \ref{eq:rel1} and \ref{eq:rel2} yields
\begin{equation}
f(S_t^{(i+1)}) - f(S_t^{(i)}) \geq \frac{\eta_t^{(i+1)}}{kc}\left(f(O_t)-f(S_t^{(i)})\right).
\end{equation}
Using a similar inductive argument as we did in the proof of Theorem \oldref{thm:card} and due to the fact that $f(\emptyset) = 0$,
\begin{equation}\label{eq:pacb1}
\begin{aligned}
f(S_t) &\geq  \left(1-\left(1- \sum_{i=1}^k\frac{\eta_t^{(i)}}{kc}\right)\right)f(O_t)\\
&\stackrel{(a)}{\geq} \left(1- e^{-\sum_{i=1}^k\frac{\eta_t^{(i)}}{kc}}\right)f(O_t),
\end{aligned}
\end{equation}
where to obtain $(a)$ we use the fact that $(1+x)^y\leq e^{xy}$ for $y>0$. To obtain the stated result, we apply the Bernstein's inequality \cite{hogg1995introduction} on the term $\sum_{i=1}^k\eta_t^{(i)}$ that is a sum of independent random variables. Note that since $\{\eta_t^{(i)}\}$ are bounded random variables, by Popoviciu's inequality \cite{hogg1995introduction} for all $i\in[k]$ it holds that 
\begin{equation}
\Var[\eta_t^{(i)}] \leq \frac{1}{4}(1-\ell_{i}(\epsilon))^2.
\end{equation}
Hence, by Bernstein's inequality for all $0<q<1$
\begin{equation}
\Pr\{\sum_{i=1}^k\eta_t^{(i)}< (1-q)\sum_{i=1}^k\mu_i\} <p
\end{equation}
where 
\begin{equation}
\begin{aligned}
p&=\exp\left(-\frac{(1-q)^2(\sum_{i=1}^k\mu_i(\epsilon))^2}{\frac{1-q}{3}\sum_{i=1}^k\mu_i(\epsilon)+\frac{1}{4}\sum_{i=1}^k(1-\ell_{i}(\epsilon))^2}\right)\\
&\stackrel{(a)}{\leq} \exp\left(-\frac{k(1-q)^2\mu_{min}^2(\epsilon)}{\frac{1-q}{3}\mu_{min}(\epsilon)+\frac{1}{4}(1-\ell_{min}(\epsilon))^2}\right)\\
&\stackrel{}{=} e^{-C(\epsilon,q)k}
\end{aligned}
\end{equation}
where $(a)$ follows as $p$ only increases by lower bounding $\mu_i(\epsilon)$ and $\ell_{i}(\epsilon)$. Finally, employing this results in \ref{eq:pacb1} yields 
\begin{equation}
f(S_t) \geq\left(1- e^{-\frac{(1-q)\mu_{min}(\epsilon)}{c}}\right)f(O_t),
\end{equation}
with probability at least $1-e^{C(\epsilon,q)k}$. This completes the proof.
\end{proof}
Indeed, in simulation studies (see Section \oldref{sec:sim}) we empirically verify the results of Theorems \oldref{thm:card} and \oldref{thm:pac} and illustrate that Algorithm 2 performs favorably compared to the competing schemes both on average and for each individual sensor scheduling tasks.

Similar to Corollary \oldref{col:mse}, we now obtain a probabilistic bound on the  achievable mean-square error in \ref{eq:mse} at each time step using the proposed randomized greedy algorithm, as stated in Corollary \oldref{col:pac} below.

\begin{corollary}\label{col:pac}
	Instate the notation and hypotheses of Corollary \oldref{col:mse} and Theorem \oldref{thm:pac}. Let $0<q<1$ and define $\alpha = 1- \exp(-\frac{(1-q)\mu_{min}(\epsilon)}{c})$. Then, with probability at least $1-e^{-Ck}$ it holds that
	\begin{equation}\label{eq:pac1}
\mathrm{MSE}_{S_t}\leq\alpha \mathrm{MSE}_{o} + (1-\alpha) \mathrm{Tr}(\P_{t|t-1}),
	\end{equation}
	for some $C>0$.
\end{corollary}
\section{Simulation Results}\label{sec:sim}
To test the performance of the proposed randomized greedy algorithm, we compare it with the classic greedy algorithm and the SDP relaxation in a variety of settings as detailed below. 

We consider the problem of Kalman filtering for state estimation in a linear time-varying system. For simplicity, let us assume that the system is in steady state and $\mathbf{H}  = \I_m$. The initial state is a zero-mean Gaussian random vector with covariance $\mathbf{\Sigma_x} = \I_m$. We further specify zero-mean Gaussian process and measurement noises with covariance matrices $\mathbf{Q}=0.05\I_m$ and $\mathbf{R}=0.05\I_n$, respectively. At each time step, the measurement vectors, i.e., the rows of the measurement matrix $\A(t)$, are drawn according to $\mathcal{N} \sim (0,\frac{1}{m}\I_m)$.

The MSE values and running time of each scheme is averaged over 10 Monte-Carlo simulations. The time horizon for each run is $T = 10$. The greedy and randomized greedy algorithms are implemented in MATLAB while the SDP relaxation scheme is implemented via CVX \cite{grant2008cvx}. All experiments were run on a laptop with 2.0 GHz Intel Core i7-4510U CPU and 8.00 GB of RAM.

We first consider the system having state dimension $m = 50$, the number of measurements $n = 400$, and $k = 55$, and compare the MSE values of each method over the time horizon of interest. For randomized greedy we set $\epsilon = 0.001$. Fig. \oldref{fig:overtime} shows that the greedy method consistently yields the lowest MSE while the MSE of the randomized greedy algorithm is slightly higher. The MSE performance achieved by the SDP relaxation is considerably larger than those of the greedy and randomized greedy algorithms. The running time of each method is given in Table \oldref{tab:sim1_time}. Both the greedy algorithm and the randomized greedy algorithm are much faster than the SDP formulation. The randomized greedy scheme is nearly two times faster than the greedy method.
\begin{table}[!htb]
\centering
\begin{tabular}{|c|c|c|}
\hline
Randomized Greedy & Greedy & SDP Relaxation \\ 
\hline
0.20 s & 0.38 s & 249.86 s\\
\hline
\end{tabular}
\caption{Running time comparison of randomized greedy, greedy, and SDP relaxation sensor selection schemes ($m=50$, $n=400$, $k=55$, $\epsilon=0.001$).}
\label{tab:sim1_time}
\end{table}
\begin{figure}[t]
\centering
    \includegraphics[width=0.49\textwidth]{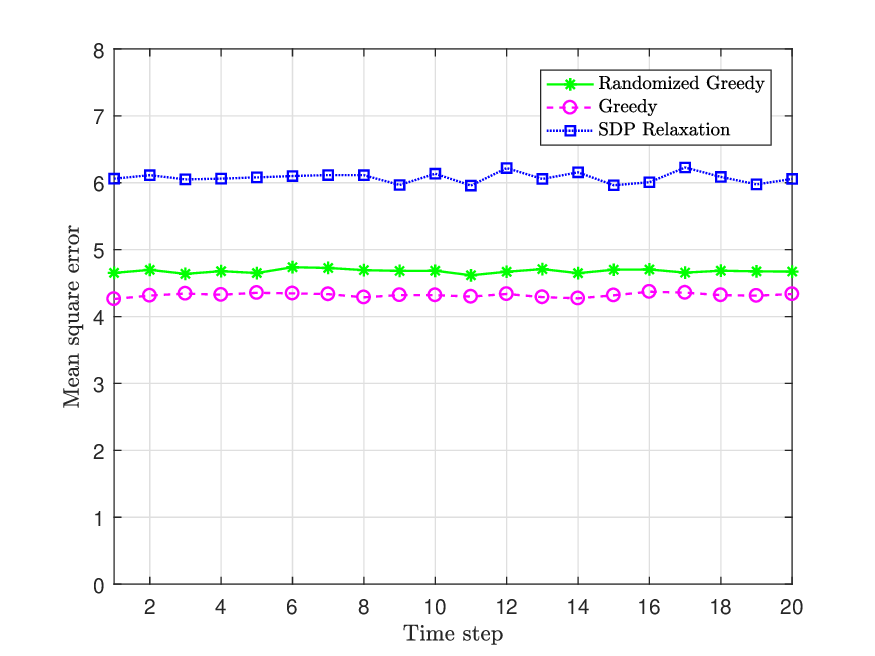}
    \caption{MSE comparison of randomized greedy, greedy, and SDP relaxation sensor selection schemes employed in Kalman filtering.}
\label{fig:overtime}
\vspace{-0.3cm}
\end{figure}
Note that, in this example, in each iteration of the sensor selection procedure the randomized scheme only computes the marginal gain for a sampled subset of size 50. As a comparison, the greedy approach computes the marginal gain for all 400 sensors. In summary, the greedy method yields the lowest MSE but is much slower than the proposed randomized greedy algorithm.

To study the effect of the number of selected sensors on performance, we vary $k$ from 55 to 115 with increments of 10. The MSE values at the last time step (i.e., $t = 10$) for each algorithm are shown in Fig. \oldref{fig:kvary}(a). As the number of selected sensors increases, the estimation becomes more accurate, as reflected by the MSE of each algorithm. Further, the difference between the MSE values consistently decreases as more sensors are selected. The running times shown in Fig. \oldref{fig:kvary}(b) indicate that the randomized greedy scheme is nearly twice as fast as the greedy method, while the SDP method is orders of magnitude slower than both greedy and randomized greedy algorithms. 
\begin{figure}[t]
\centering
    \begin{subfigure}{.49\textwidth}
  \centering
    \includegraphics[width=1\textwidth]{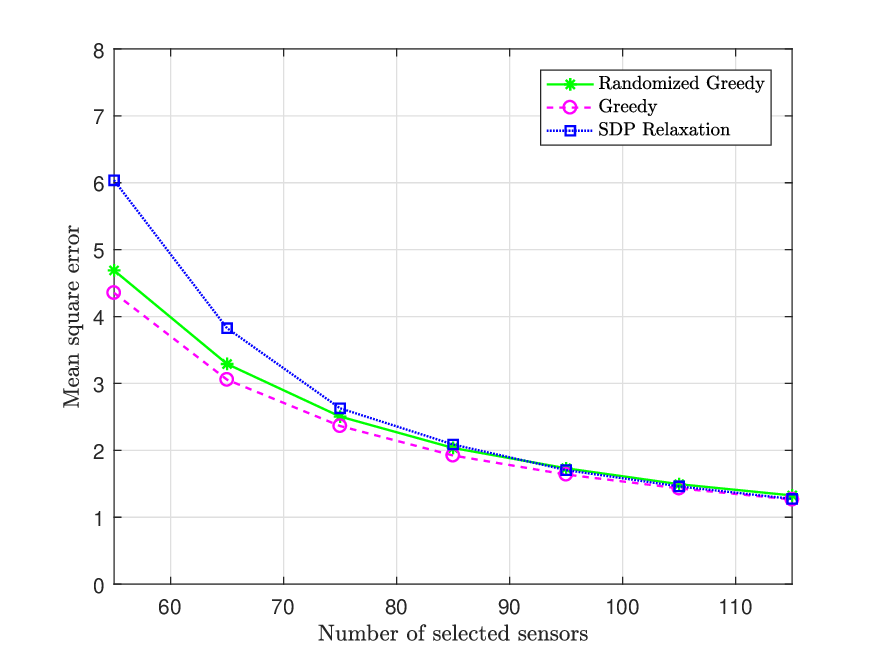}\caption{\scriptsize MSE comparison}
        \end{subfigure}
        \begin{subfigure}{.49\textwidth}
  \centering
    \includegraphics[width=1\textwidth]{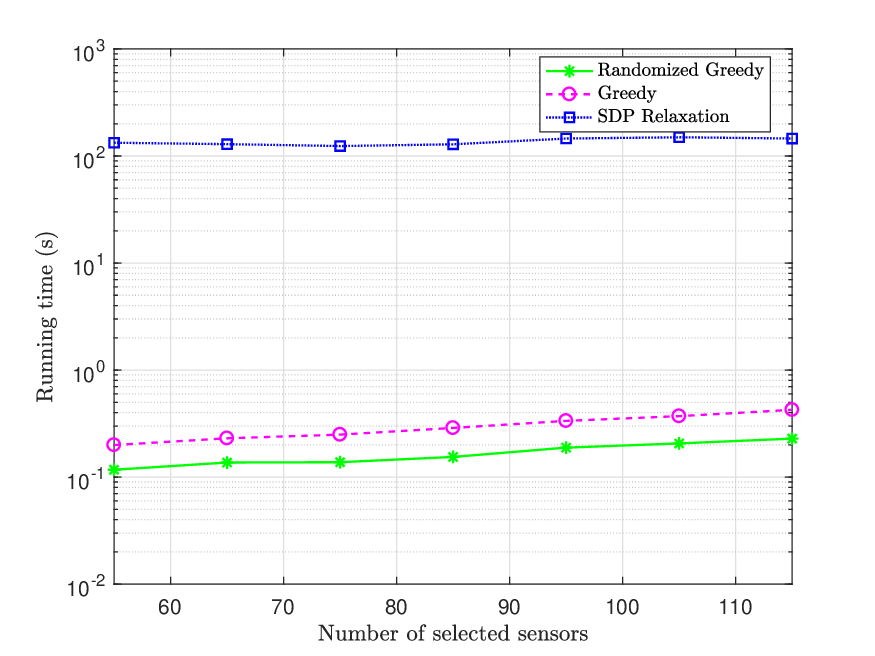}\caption{\scriptsize Running time comparison}
\end{subfigure}
\caption{Comparison of randomized greedy, greedy, and SDP relaxation schemes as the number of selected sensors increases.}
\label{fig:kvary}
\vspace{-0.3cm}
\end{figure}

Finally, we compare the performance of the randomized greedy algorithm to that of the greedy algorithm as the size of the system increases. We run both methods for 20 different sizes of the system. The initial size was set to $m = 20$, $n=200$, and $k=25$ and all three parameters are scaled by $\beta$ where $\beta$ varies from 1 to 20. In addition, to evaluate the effect of $\epsilon$ on the performance and runtime of the randomized greedy approach, we repeat experiments for $\epsilon \in \{0.1, 0.01, 0.001\}$.
Note that the computational complexity of the SDP relaxation scheme is prohibitive in this setting. 
Fig. \oldref{fig:betavary}(a) illustrates the percentage difference of the MSE between the two methods. In particular, we show 
\begin{equation*}
\% \: \Delta \mathrm{MSE} = \frac{\mathrm{MSE_{RG}}-\mathrm{MSE_{G}}}{\mathrm{MSE_{G}}} \times 100 
\end{equation*}
where ``RG'' and ``G''  refer to the randomized greedy and greedy algorithms, respectively. It can be seen that this difference between the MSEs reduces as the system scales up.
The running time is plotted in Fig. \oldref{fig:betavary}(b). As the figure illustrates, the gap between the running times grows with the size of the system and the randomized greedy algorithm performs nearly 25 times faster than the greedy method for the largest network.
Fig. \oldref{fig:betavary} shows that using a smaller $\epsilon$ results in a lower MSE  while it slightly increases the running time. These results suggest that, for large systems, the randomized greedy provides almost the same MSE while being much faster than the greedy algorithm.
\begin{figure}[t]
\centering
    \begin{subfigure}{.49\textwidth}
  \centering
    \includegraphics[width=1\textwidth]{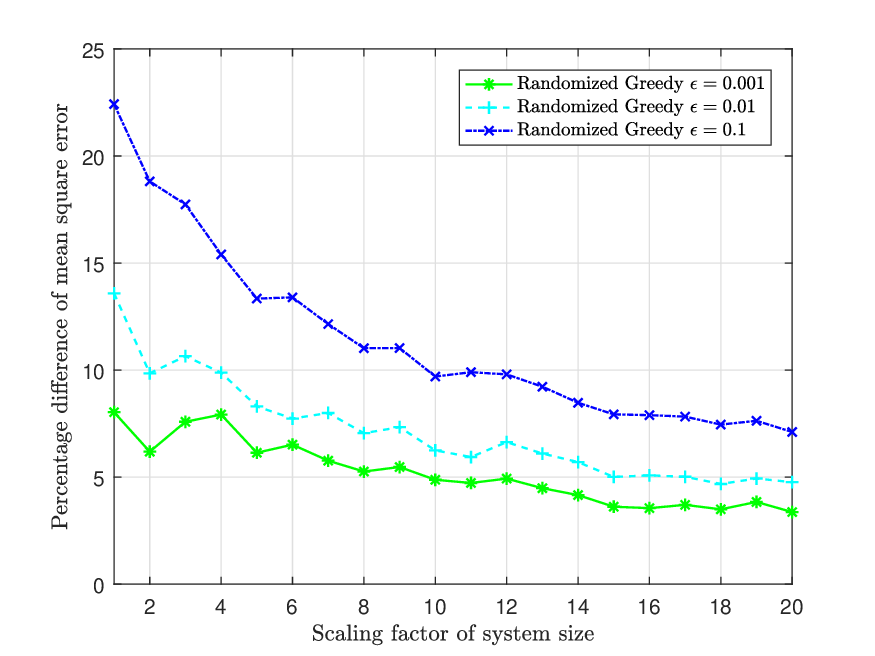}\caption{\scriptsize $ \Delta \: \mathrm{MSE}$}
        \end{subfigure}
        \begin{subfigure}{.49\textwidth}
  \centering
    \includegraphics[width=1\textwidth]{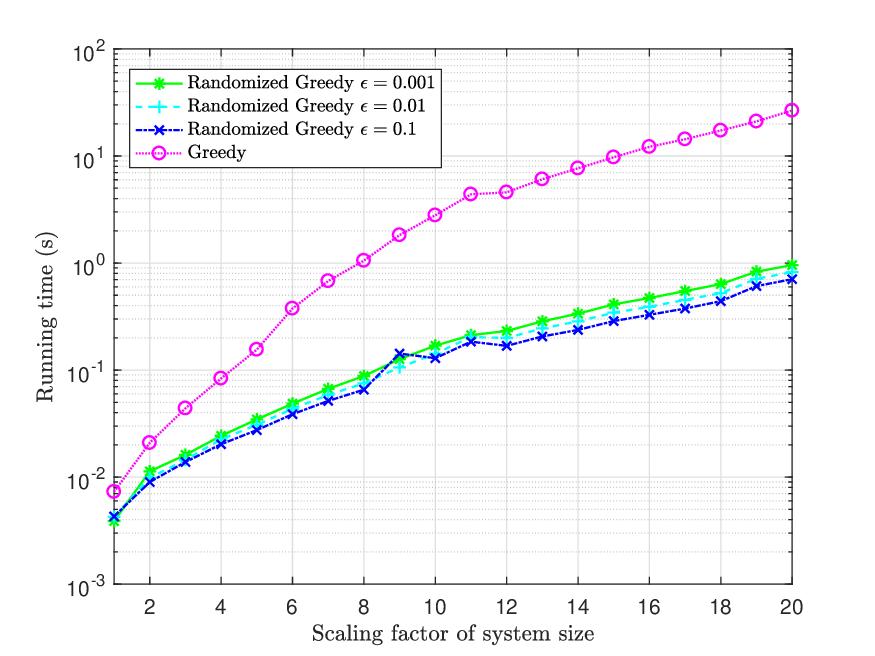}\caption{\scriptsize Running time comparison}
\end{subfigure}
\caption{A comparison of the randomized greedy and greedy algorithms for varied network size.}
\label{fig:betavary}
\vspace{-0.3cm}
\end{figure}
\section{Conclusion} \label{sec:concl}
In this paper, we considered the problem of state estimation in large-scale linear time-varying dynamical systems. We proposed a randomized greedy algorithm for selecting sensors to query such that their choice minimizes the estimator's mean-square error at each time step. We established the performance guarantee for 
the proposed algorithm and analyzed its computational complexity.
To our knowledge, the proposed scheme is the first randomized algorithm for sensor scheduling with an explicit bound on its achievable mean-square error.
In addition, we provided a probabilistic theoretical bound on the element-wise curvature of the objective function. Furthermore, in simulations we demonstrated that the proposed algorithm is superior to the classical greedy and SDP relaxation methods in terms of running time while providing the same or better utility.

As a future work, we intend to extend this approach to nonlinear dynamical systems and obtain a theoretical guarantee on the quality of the resulting approximate solution found by randomized greedy algorithm. Moreover, it would be of interest to extend the framework established in this manuscript to related problems such as minimal actuator placement.
\normalsize
\begin{appendices}
\section{Derivation of SDP Relaxation of \ref{sdp}}\label{pf:sdp}
Let $z_i \in \{0,1\}$ indicate the
membership of the $i\ts{th}$ sensor in the selected subset at time $t$ and define $\z = [z_1,z_2,\dots,z_n]^\top$. Hence, \ref{eq:sensor} can be written as
\begin{equation*}\label{sensor1}
\begin{aligned}
& \underset{\z}{\text{min}}
\quad \mathrm{Tr}\left(\left(\P_{t|t-1}^{-1}+\sigma^{-2}\sum_{i=1}^n z_i\a_i(t)\a_i(t)^\top\right)^{-1}\right)\\
& \text{s.t.}\hspace{0.5cm}  z_i \in \{0,1\}, \phantom{k} \forall i\in [n]\\
& \hspace{0.9cm} \sum_{i=1}^n z_i =k.
\end{aligned}
\end{equation*}
The convex relaxation of the above optimization problem is
given by
\begin{equation}\label{sensor2}
\begin{aligned}
& \underset{\z}{\text{min}}
\quad \mathrm{Tr}\left(\left(\P_{t|t-1}^{-1}+\sigma^{-2}\sum_{i=1}^n z_i\a_i(t)\a_i(t)^\top\right)^{-1}\right)\\
& \text{s.t.}\hspace{0.5cm}  0\leq z_i \leq 1, \phantom{k} \forall i\in [n]\\
& \hspace{0.9cm} \sum_{i=1}^nz_i =k.
\end{aligned}
\end{equation}
In order to obtain an SDP in standard form, let $\Y$ be a positive
semidefinite matrix such that 
\begin{equation}\label{shur}
\Y \succeq \left(\P_{t|t-1}^{-1}+\sigma^{-2}\sum_{i=1}^n z_i\a_i(t)\a_i(t)^\top\right)^{-1}
\end{equation}
Then, \ref{sensor2} can equivalently be written as 
\begin{equation}\label{sensor3}
\begin{aligned}
& \underset{\z,\Y}{\text{min}}
\quad \mathrm{Tr}(\Y)\\
& \text{s.t.}\hspace{0.5cm}  0\leq z_i \leq 1, \phantom{k} \forall i\in [n]\\
& \hspace{0.9cm} \sum_{i=1}^nz_i =k\\
& \hspace{1cm} \Y - \left(\P_{t|t-1}^{-1}+\sigma^{-2}\sum_{i=1}^n z_i\a_i(t)\a_i(t)^\top\right)^{-1}\succeq \mathbf{0}
\end{aligned}
\end{equation}
Note that the expression on the left hand side of last constraint in \ref{sensor3} can be thought of as the Schur complement \cite{horn2012matrix} of the block PSD matrix
\begin{equation}
\B = \begin{bmatrix}
\Y& \I\\
\I & \P_{t|t-1}^{-1}+\sigma^{-2}\sum_{i=1}^n z_i\a_i(t)\a_i(t)^\top
\end{bmatrix}.
\end{equation}
Since Schur complement of $\B$ is positive semidefinite if and
 only if $\B \succeq \mathbf{0}$, we obtain the SDP relaxation given in \ref{sdp}.

The solution to the SDP may take fractional
values, in which case some kind of sorting and rounding need
to be employed in order to obtain the desired solution. Here, we select the sensors corresponding to the $k$ $z_i$'s with largest values, as originally suggested by \cite{joshi2009sensor}.
\section{Proof of Proposition \oldref{p:mono}}\label{pf:mono}
First, note that 
\[f(\emptyset)=\mathrm{Tr}\left(\P_{t|t-1}-\F_{\emptyset}^{-1}\right)=\mathrm{Tr}\left(\P_{t|t-1}-\P_{t|t-1}\right)=0.\]
Now, for $j\in[n]\backslash S$ it holds that
\begin{equation}\label{eq:update}
\begin{aligned}
f_j(S) &= f(S\cup \{j\}) -f(S) \\
&= \mathrm{Tr}\left(\P_{t|t-1}-\F_{S\cup \{j\}}^{-1}\right) -  \mathrm{Tr}\left(\P_{t|t-1}-\F_S^{-1}\right)\\
&= \mathrm{Tr}\left(\F_S^{-1}\right) -\mathrm{Tr}\left(\F_{S\cup \{j\}}^{-1}\right)\\
&=  \mathrm{Tr}\left(\F_S^{-1}\right)-\mathrm{Tr}\left(\left(\F_S+\sigma^{-2}\a_j(t)\a_j(t)^\top\right)^{-1}\right)\\
&\stackrel{(a)}{=} 
\mathrm{Tr}\left(\frac{\F_S^{-1}\a_j(t)\a_j(t)^\top\F_S^{-1}}{\sigma^{2}+\a_j(t)^\top\F_S^{-1}\a_j(t)}\right)  \\
&\stackrel{(b)}{=} \frac{\a_j(t)^\top\F_S^{-2}\a_j(t)}{\sigma^{2}+\a_j(t)^\top\F_S^{-1}\a_j(t)}
\end{aligned}
\end{equation}
where $(a)$ is by applying matrix inversion lemma (Sherman–Morrison formula) \cite{bellman1997introduction} on $(\F_S+\sigma^{-2}\a_j(t)\a_j(t)^\top)^{-1}$, and $(b)$ is by properties of trace of a matrix. Finally, since $\F_S$ is a symmetric positive definite matrix, $f_j(S) > 0$ which in turn implies monotonicity.
\section{Proof of Lemma \oldref{lem:curv}}\label{pf:curv}
Let $S \subset T$ and $T \backslash S=\{j_1,\dots,j_r\}$. Therefore,
\begin{multline}
f(T)-f(S)=f(S\cup \{j_1,\dots,j_r\})-f(S)\\
\hspace{0.5cm}={f_{j_1}(S)}+{f_{j_2}(S\cup\{j_1\})}+\dots\\+{f_{j_r}(S\cup\{j_1,\dots,j_{r-1}\})}.
\end{multline}
Applying definition of element-wise curvature yields  
\begin{equation}\label{ds}
\begin{aligned}
f(T)-f(S)&\leq {f_{j_1}(S)}+{\cal C}_1{f_{j_2}(S)}+\dots+{\cal C}_{r-1}{f_{j_r}(S)}\\
&={f_{j_1}(S)}+\sum_{l=1}^{r-1}{\cal C}_l{f_{j_t}(S)}.
\end{aligned}
\end{equation}
Note that \ref{ds} is invariant to the ordering of elements in $T \backslash S$. In fact, it is straightforward to see that given ordering $\{j_1,\dots,j_r\}$, one can choose a set $P=\{{\cal P}_1,\dots,{\cal P}_r\}$ with $r$ permutations -- e.g., by defining the right circular-shift operator ${\cal P}_t(\{j_1,\dots,j_r\})=\{j_{r-t+1},\dots,j_1,\dots\}$ for $1 \leq t\leq r$ -- such that ${\cal P}_p(j)\neq{\cal P}_q(j)$ for $p\neq q$ and $\forall j\in T \backslash S$. Hence, \ref{ds} holds for $r$ such permutations. Summing all of these $r$ inequalities we obtain
\begin{equation}\label{qs}
r(f(T)-f(S))\leq \left(1+\sum_{l=1}^{r-1}{\cal C}_l\right)\sum_{j\in T \backslash S}{f_{j}(S)}.
\end{equation}
Rearranging \ref{qs} yields the desired result.
\section{Proof of Lemma \oldref{lem:rand}}\label{pf:rand}
First, we aim to bound the probability of the event that the random set $R$ contains at least an index from the optimal set of sensor as this is a necessary condition to reach the optimal MSE. Consider $S_t^{(i)}$, the set of selected sensors at the end of $i\ts{th}$ iteration of Algorithm \oldref{alg:greedy} and let $\Phi = R \cap (O\backslash S_t^{(i)})$. It holds that\footnote{Note that without loss of generality and for simplicity we assume that $s$ is an integer.}
\begin{equation}
\begin{aligned}
\Pr\{\Phi = \emptyset\}& = \prod_{l = 0}^{s-1} \left(1-\frac{|O\backslash S_t^{(i)}|}{|[n]\backslash S_t^{(i)}|-l}\right)\\
&\stackrel{(a)}{\leq}\left(1-\frac{|O\backslash S_t^{(i)}|}{s}\sum_{l=0}^{s-1}\frac{1}{|[n]\backslash S_t^{(i)}|-l}\right)^s\\
&\stackrel{(b)}{\leq}\left(1-\frac{|O\backslash S_t^{(i)}|}{s}\sum_{l=0}^{s-1}\frac{1}{n-l}\right)^s
\end{aligned}
\end{equation}
where $(a)$ is by the inequality of arithmetic and geometric means, and $(b)$ holds since $|[n]\backslash S_i|\leq n$. Now recall for any integer $p$,
\begin{equation}\label{eq:har}
H_p=\sum_{l=1}^p\frac{1}{p}=\log p + \gamma+\zeta_p
\end{equation}
where, $H_p$ is the $p\ts{th}$ harmonic number, $\gamma$ is the Euler–Mascheroni constant, and $\zeta_p	 = \frac{1}{2p} - \mathcal{O}(\frac{1}{p^4})$ is a monotonically decreasing sequence related to Hurwitz zeta function \cite{lang2013algebraic}. Therefore, using the identity \ref{eq:har} we obtain
\begin{equation}
\begin{aligned}
\Pr\{\Phi = \emptyset\}&\stackrel{}{\leq} (1-\frac{|O\backslash S_t^{(i)}|}{s}(H_n-H_{n-s}))^s \\
&\stackrel{}{=} (1-\frac{|O\backslash S_t^{(i)}|}{s}(\log(\frac{n}{n-s})+\zeta_n-\zeta_{n-s}))^s\\
&\stackrel{(a)}{\leq} (1-\frac{|O\backslash S_t^{(i)}|}{s}(\log(\frac{n}{n-s})-\frac{s}{2n(n-s)}))^s\\
&\stackrel{(b)}{\leq} ((1-\frac{s}{n})e^{\frac{s}{2n(n-s)}})^{|O\backslash S_t^{(i)}|}
\end{aligned}
\end{equation}
where $(a)$ follows since $\zeta_n-\zeta_{n-s} = \frac{1}{2n} -\frac{1}{2(n-s)} + \mathcal{O}(\frac{1}{(n-s)^4})$, and $(b)$ is by the fact that $(1+x)^y\leq e^{xy}$ for any real number $y>0$. Next, the fact that $\log (1-x)\leq -x-\frac{x^2}{2}$ for $0<x<1$ yields
\begin{equation}
(1-\frac{s}{n})e^{\frac{s}{2n(n-s)}} \leq e^{-\frac{\beta_1 s}{n}}
\end{equation}
where $\beta_1 = 1 +(\frac{s}{2n}-\frac{1}{2(n-s)})$. On the other hand, we can also upper bound $\Pr\{\Phi = \emptyset\}$ as
\begin{equation}
\begin{aligned}
\Pr\{\Phi = \emptyset\}& \leq \left(1-\frac{|O\backslash S_t^{(i)}|}{s}\sum_{l=0}^{s-1}\frac{1}{n-l}\right)^s\\
& \leq\left(1-\frac{|O\backslash S_t^{(i)}|}{n}\right)^s\\
&\leq e^{-\frac{s}{n}|O\backslash S_t^{(i)}|}
\end{aligned}
\end{equation}
where we again employed the inequality  $(1+x)^y\leq e^{xy}$. Now, let $\beta = \max\{1,\beta_1\}$. Thus, 
\begin{equation}\label{eq:pbound}
\Pr\{\Phi \neq \emptyset\} \geq 1- e^{-\frac{\beta s}{n}|O\backslash S_t^{(i)}|}\geq \frac{1-\epsilon^\beta}{k}(|O\backslash S_t^{(i)}|)
\end{equation}
by definition of $s$ and the fact that $1- e^{-\frac{\beta s}{n}x} $ is a concave function. Finally, according to Lemma 2 in \cite{mirzasoleiman2014lazier},
\begin{equation}\label{eq:lazy}
\E[f_{(i+1)_s}(S_t^{(i)})|S_t^{(i)}]\geq \frac{\Pr\{\Phi \neq \emptyset\}}{|O\backslash S_t^{(i)}|}\sum_{j\in O_t\backslash S_t^{(i)}}f_o(S_t^{(i)}).
\end{equation}
Combining \ref{eq:pbound} and \ref{eq:lazy} yields the stated results.
\end{appendices}
\bibliographystyle{IEEEtran}
\bibliography{IEEEabrv,refs}
\end{document}